\documentclass[11pt]{article}
\usepackage{amsmath,amsfonts,amsthm,amssymb,ifthen}

\newboolean{@laptop}
\newboolean{@todoon}
\setboolean{@laptop}{true}
\setboolean{@todoon}{true}

\ifthenelse{\boolean{@laptop}}{}{\usepackage{mathbbol}}

\def\todo#1{\ifthenelse{\boolean{@todoon}}{\marginpar{\textit{#1}}}{}}

\newtheorem{theorem}{Theorem}[section]

\newtheorem{lemma}[theorem]{Lemma}

\newtheorem{proposition}[theorem]{Proposition}
\newtheorem{definition}[theorem]{Definition}

\newcommand{\setj}[2]{S_{#1} (#2)}
\newcommand{\lamj}[2]{\lambda_{#1} (#2)}

{\begin{list}{}{\usecounter{bean}%
  \setlength{\topsep}{0em}
  \setlength{\parsep}{0em}
  \setlength{\partopsep}{0em}
  \setlength{\itemsep}{0em}
}}%
{\end{list}}

\def\round#1{\left[#1 \right]_{\epsilon }}

\def\form#1#2{#1^{T} #2}

\def\myPhiSym{\ifthenelse{\boolean{@laptop}}%
{\varphi}
{\mathbb{\Phi}}
}

\def\balance#1{\textbf{bal}\left(#1\right)}

\def\union{\cup}
\def\intersect{\cap}

\def\defeq{\stackrel{\mathrm{def}}{=}}

\def\prob#1#2{\Pr_{#1}\left[ #2 \right]}
\def\expec#1#2{\mbox{\bf E}_{#1}\left[ #2 \right]}

\def\norm#1{\left\| #1 \right\|}

\def\infnorm#1{\left\| #1 \right\|_{\infty }}

\def\setof#1{\left\{#1  \right\}}
\def\sizeof#1{\left|#1  \right|}

\def\bvec#1{{\mbox{\boldmath $#1$}}}

\def\intersect{\cap}

\newcommand{\ceiling}[1]{\left\lceil#1\right\rceil}

\def\setof#1{\left\{#1  \right\}}

\newdimen\pIR
\newcommand\StevesR{{\rm I\kern\pIR R}}

\def\conduc#1#2{\Phi_{#1}\left(#2  \right)}
\def\conducin#1#2{\Phi^{G}_{#1}\left(#2  \right)}

\def\Conduc#1{\Phi_{#1}}
\def\Conducin#1{\Phi^{G}_{#1}}

\def\vol#1{\mu \left(#1  \right)}

\begin{document}

\title{A Local Clustering Algorithm for Massive Graphs and its Application to Nearly-Linear Time Graph
   Partitioning\thanks{%
This paper is the first in a sequence of three papers expanding
  on material that appeared first under the title
  ``Nearly-linear time algorithms for graph partitioning, 
    graph sparsification, and solving linear systems''~\cite{SpielmanTengPrecon}.
The second paper, ``Spectral Sparsification of Graphs''~\cite{SpielmanTengSparsifier}
  contains further results on partitioning graphs, and applies them to producing
  spectral sparsifiers of graphs.
The third paper, ``Nearly-Linear Time Algorithms for Preconditioning and Solving Symmetric, Diagonally Dominant Linear Systems''~\cite{SpielmanTengLinsolve} contains the results
  on solving linear equations and approximating eigenvalues and eigenvectors.
\vskip 0.01in
This material is based upon work supported by the National Science Foundation 
  under Grant Nos. 0325630, 0634957, 0635102 and 0707522.
Any opinions, findings, and conclusions or recommendations expressed in this material are those of the authors and do not necessarily reflect the views of the National Science Foundation.
}
}
\author{
Daniel A. Spielman\\
Department of Computer Science\\
Program in Applied Mathematics\\
Yale University
\and
Shang-Hua Teng\\
Department of Computer Science\\
Boston University}

\maketitle

\begin{abstract}
We study the design of {\em local algorithms} for 
  massive graphs.
A local algorithm is one that
  finds a solution containing or near a given vertex without
  looking at the whole graph.
We present a local clustering algorithm.
Our algorithm finds a good cluster---a subset of vertices
  whose internal connections are significantly richer
  than its external connections---near a given vertex.
The running time of our algorithm, when it finds a non-empty
  local cluster, is nearly linear in the size
  of the cluster it outputs.

Our clustering algorithm could be  a useful primitive
  for handling massive graphs, such as social networks and web-graphs.
As an application of this clustering algorithm,
  we present a partitioning algorithm
  that finds an approximate sparsest cut with nearly optimal balance.
Our algorithm takes time nearly linear in the number edges of the graph.

Using the partitioning algorithm of this paper, we have designed a nearly-linear
  time algorithm for constructing spectral sparsifiers of graphs, which
  we in turn use in a nearly-linear time algorithm for solving linear
  systems in symmetric, diagonally-dominant matrices.
The linear system solver also leads to a nearly linear-time
  algorithm for approximating 
  the second-smallest eigenvalue and corresponding eigenvector
  of the Laplacian matrix of a graph.
These other results are presented in two companion papers.
\end{abstract}

\newpage
\section{Introduction}\label{sec:Intro}

Given a vertex of interest in a massive graph, we
  would like to find a small cluster around that vertex, 
  \textit{in time proportional to the size of the cluster}.
The algorithm we introduce will solve this problem while
  only examining vertices near the initial
  vertex, under some reasonable notion of nearness.
We call such an algorithm a \textit{local} algorithm.

Our local clustering algorithm provides a very powerful
  primitive for the design of fast graph algorithms.
In Section~\ref{sec:cut} of this paper, we use it to design
  the first nearly-linear time algorithm for graph partitioning
  that produces a partition of nearly-optimal balance among those
  approximating a target conductance.
In the papers~\cite{SpielmanTengSparsifier} and~\cite{SpielmanTengLinsolve},
  we proceed to use this graph partitioning algorithm to design
  nearly-linear time algorithms for sparsifying graphs and for solving
  symmetric, diagonally-dominant linear systems.

\subsection{Local Clustering}
We say that a graph algorithm is a \textit{local algorithm} if it is given
  a particular vertex as input,
  and at each step after the first only examines vertices connected to those
  it has seen before.
The use of a local algorithm naturally leads to the question of in which
  order one should explore the vertices of a graph.
While it may be natural to explore vertices in order of shortest-path
  distance from the input
  vertex, such an ordering is a poor choice in graphs of low-diameter,
  such as social network graphs~\cite{diameterSocialNetwork}.
We suggest first processing the vertices that are most likely
  to occur in short random walks from at the input vertex.
That is, we consider a vertex to be near the input vertex if it is likely
  to appear in a short random walk from the input vertex.

In Section~\ref{sec:cut}, we use a local graph exploration process
  to find a cluster that is near the input vertex.
Following Kannan, Vempala and Vetta~\cite{KannanVempalaVetta}, we say that
  a set of vertices is a good cluster if
  it has low \textit{conductance}; that is,
  if it has many more external than internal edges.
We give an efficient local clustering algorithm, \texttt{Nibble}, that
  runs in time proportional to the size of the cluster it outputs.
Although our algorithm may not find a local cluster
  for some input vertices,
  we will show that it is usually successful.
In particular, we prove the following theorem:
There exists a constant $\alpha > 0$ such that for any
   target conductance $\phi$ and any cluster $C_0$ 
   of conductance at most $\alpha\cdot\phi^2/\log^3 n$,
   when given a random vertex $v$  sampled according to degree
   inside $C_0$,
   \texttt{Nibble} will return a cluster $C$ mostly inside $C_0$  and with
  conductance at most $\phi$,
  with probability at least $1/2$.

The local clustering algorithm \texttt{Nibble} makes a novel use of random walks.
For a positive integer $t$, suppose $p_{t,v}$ is
  the probability distribution of the $t$-step
  random walk starting at $v$.
As the support of $p_{t,v}$---the set of nodes with positive probability---could
  grow rapidly, \texttt{Nibble} maintains a truncated version
  of the distribution.
At each step of the truncated random walks,
  \texttt{Nibble} looks a for cluster among only nodes
  with high probability.
The truncation is critical to ensure that
    the clustering algorithm is output sensitive.
It guarantees that the size of the support of the distribution that
  \texttt{Nibble}
  maintains is not too much larger than the size of the cluster it produces.
The cluster that \texttt{Nibble}  produces is local to the starting
  vertex $v$ in the sense that it consists of nodes that are among
  the most favored destinations of random walks starting from $v$.

By using the personal PageRank vector~\cite{PageRank}
  to define nearness, Andersen, Chung and Lang~\cite{AndersenChungLang}, 
  have produced an improved version of our algorithm \texttt{Nibble},
  which they call \texttt{PageRank-Nibble}.
Following this work, other local algorithms have been designed by
  Andersen \textit{et. al.}~\cite{AndersenPageRank} for
  approximately computing Personal PageRank vectors, by
  Andersen~\cite{AndersenDense} for finding dense subgraphs
  and by Andersen, Chung and Lang~\cite{AndersenChungLang2} for partitioning
  directed graphs.

\subsection{Nearly Linear-Time Algorithms}

Our local clustering algorithm provides a powerful tool for
  designing fast graph algorithms.
In this paper and its two companion papers, we 
  show how to use it to design randomized, nearly linear-time
  algorithms for several important graph-theoretic and
  numerical problems.

The need for algorithms whose running time is
  linear or nearly linear in their input size
  has increased as algorithms handle larger inputs.
For example, in circuit design and simulation,
  an Intel Dual Core Itanium processor has more than
  one billion transistors, 
  which is  more than 100 times the number of
  transistors that the Pentium had in 2000~\cite{IntelMooresLaw};
  in scientific computing, one often needs to solve linear
  systems that involve hundreds of millions of variables~\cite{LargeScaleScientificComputing};
 in modern information infrastructure, the web has grown into
  a graph of hundreds billions of nodes~\cite{IndexableWeb2005}.
As a result of this rapid growth in problem size, what used to
  be considered an efficient algorithm, such as
  a $O(n^{1.5})$-time algorithm, may no longer be adequate for
  solving problems of these scales.
Space complexity poses an even greater problem.

Many basic graph-theoretic problems such as connectivity
  and topological sorting can be solved in linear or nearly-linear time.
The efficient algorithms for these problems are built on
  linear-time primitives such as Breadth-First-Search
  (BFS) and Depth-First-Search (DFS).
Minimum Spanning Trees (MST) and Shortest-Path Trees are
  examples of other commonly used nearly linear-time primitives.
We hope to build up the library of nearly-linear time
  graph algorithms that may be used as primitives.
While the analyzable variants of the 
  algorithms we present here, and even their improved versions
  by Andersen, Chung and Lang~\cite{AndersenChungLang}, may not be
  immediately useful in practice, we 
  believe practical algorithms may be derived from them by
  making less conservative choices of parameters.


Our local clustering algorithm provides an exciting new
  primitive for developing nearly linear-time graph algorithms.
Because its running time is proportional to the size of
  the cluster it produces, we can repeatedly apply it
  remove many clusters from a graph, all within nearly-linear time.

In the second part of this paper,
  we use  \texttt{Nibble}  as a subroutine to
  construct a randomized graph partitioning
  algorithm that runs in nearly-linear time.
To the best of our knowledge, this is the first nearly linear-time
  partitioning algorithm that finds an approximate sparsest cut
  with approximately optimal balance.
In our first companion paper~\cite{SpielmanTengSparsifier}, we apply this new
  partitioning algorithm to develop a
  nearly-linear-time algorithm for producing spectral sparsifiers of graphs.
We begin that paper by extending the partitioning algorithm of this paper
  to obtain a stronger guarantee on its output: if it outputs a small set,
  then the complement must be contained in a subgraph whose conductance is
  higher than the target.

\section{Clusters and Conductance}\label{sec:Def}

Let  $G = (V,E)$ be an undirected graph  with $V = \setof{1,\dotsc ,n}$.
A {\em cluster} of $G$ is a subset of $V$ that is
  richly intra-connected but sparsely connected
  with the rest of the graph.
The quality of a cluster can be measured by its conductance,
  the ratio of the number of its external connections to
  the number of its total connections.

We let $d (i)$ denote the degree of vertex $i$.
For $S \subseteq V$, we define $\vol{S} = \sum_{i \in S} d (i)$
  (often called the volume of $S$).
So, $\vol{V} = 2|E|$.
Let $E(S,V-S)$ be the set of edges connecting a vertex in
  $S$ with a vertex in $V-S$.
We define the {\em conductance} of a set of vertices $S$, written
  $\conduc{}{S}$ by
\[
  \conduc{}{S} \defeq
  \frac{\sizeof{E (S, V - S)}}
       {\min \left(\vol{S}, \vol{V - S} \right)}.
\]
The {\em conductance} of $G$ is then given by
\[
  \Conduc{G}{} \defeq \min_{S \subset V} \conduc{}{S}.
\]

We sometime refer to a subset $S$ of $V$ as a {\em cut} of $G$
  and refer to $(S,V-S)$ as a {\em partition} of $G$.
The {\em balance} of a cut $S$ or a partition $(S,V-S)$ is then equal to
  $$\balance{S} = \min (\vol{S}, \vol{V-S})/\vol{V}.$$
We call $S$ a \textit{sparsest cut}
  of $G$ if $\conduc{}{S} = \Conduc{G}{}$ and $\vol{S}/\vol{V}\leq 1/2$.

In the construction of a partition of $G$, we will be concerned with
  vertex-induced subgraphs of $G$.
However, when measuring the conductance and volumes of vertices in
  these vertex-induced subgraphs, we will continue to measure the
  volume according to the degrees of vertices in the original graph.
For clarity, we define the conductance of a set $S$ in the subgraph induced
  by $A \subseteq V$ by
\[
  \conducin{A}{S} \defeq
  \frac{\sizeof{E (S, A - S)}}
       {\min \left(\vol{S}, \vol{A - S} \right)},
\]
and
\[
  \Conducin{A}{}
   \defeq \min_{S \subset A} \conducin{A}{S}.
\]
For convenience, we define $\conducin{A}{\emptyset} = 1$ and, for $\sizeof{A} = 1$,
  $\Conducin{A}{} = 1$.

For $A \subseteq V$, we let $G (A)$ denote the subgraph of $G$ induced by
  the vertices in $A$.
We introduce the notation $G[A]$ to denote graph $G (A)$ to which self-loops
  have been added so that every vertex in $G[A]$ has the same degree
  as in $G$.
Each self-loop adds 1 to the degree.
We remark that if $G (A)$ is the subgraph of $G$ induced on the
  vertices in $A$, then
\[
  \Conducin{A}{} \leq     \Conduc{G (A)}{}.
\]
So, when we prove lower bounds on $ \Conducin{A}{}$, we obtain lower
  bounds on $ \Conduc{G (A)}{}$.

Clustering is an optimization problem:
Given an undirected graph $G$ and
   a conductance parameter, find a cluster $C$ such that
   $\conduc{}{C} \leq \phi$, or determine no such cluster exists.
The problem is  NP-complete (see, for example~\cite{LeightonRao}
  or~\cite{NPcompleteCluster}).
But, approximation algorithms exist.
Leighton and Rao~\cite{LeightonRao} used linear programming to obtain
 $O (\log n)$-approximations of the sparsest cut.
Arora, Rao and Vazirani~\cite{AroraRaoVazirani} improved this to $O (\sqrt{\log n})$
  through semi-definite programming.
Faster algorithms obtaining similar guarantees have been constructed by
  Arora, Hazan and Kale~\cite{AroraHazanKale},
  Khandekar, Rao and Vazirani~\cite{KhandekarRaoVazirani},
  Arora and Kale~\cite{AroraKale}, and
  Orecchia, Schulman, Vazirani, and Vishnoi~\cite{Orecchia}.

\subsection{The Algorithm \texttt{Nibble}}

The algorithm
  \texttt{Nibble}
  works by approximately computing the distribution
  of a few steps of the random walk starting at a
  seed vertex $v$.
It is implicit in the analysis of the volume
  estimation algorithm of Lov\'asz and
  Simonovits~\cite{LovaszSimonovits}
  that one can find a cut with small conductance from the distributions
  of the steps of the
  random walk starting at any vertex from which
  the walk does not mix rapidly.
We will observe that a random vertex in a set of low conductance
  is probably such a vertex.
We then extend the analysis of Lov\'asz and Simonovits
  to show one can find a cut with small conductance from approximations
  of these distributions, and that
  these approximations can be computed quickly.
In particular, we will truncate all
  small probabilities that appear in the distributions to 0.
In this way, we reduce the work required to
  compute our approximations.

For the rest of this section, we will work with a graph
  $G = (V,E)$ with $n$ vertices and $m$ edges, so that
  $\vol{V} = 2m$.
We will allow some of these edges to be self-loops.
Except for the self-loops, which we allow to occur with multiplicities,
  the graph is assumed to be unweighted.
We will let $A$ be the adjacency matrix of this graph.
That is,
\[
A (u,v) =
\begin{cases}
1 & \text{if $(u,v) \in E$ and $u \not = v$}
\\
k & \text{if $u = v$ and this vertex has $k$ self-loops}\\
0 & \text{otherwise}.
\end{cases}
\]

We define the following two vectors supported on a set of vertices $S$:
\begin{align*}
  \chi _{S} (u)
& =
  \begin{cases}
     1 & \text{for $u \in S$,}\\
     0 & \text{otherwise},
  \end{cases}\\
  \psi _{S} (u)
& =
  \begin{cases}
     d (u) / \vol{S} & \text{for $u \in S$,}\\
     0 & \text{otherwise}.
  \end{cases}
\end{align*}

We will consider the random walk
  that at each time step stays at the current vertex with probability $1/2$,
  and otherwise moves to the endpoint of a random edge attached
  to the current vertex.
Thus, self-loops increase the chance the walk stays at the current vertex.
For example, if a vertex has $4$ edges, one of which is a self-loop,
  then when the walk is at this vertex it has a $5/8$ chance of staying
  at that vertex, and a $1/8$ chance of moving to each of its $3$ neighbors.

The matrix realizing this walk can be expressed by
 $M = (A D^{-1} + I) / 2$,
  where
  $d (i)$ is the degree of node $i$,
  and $D$ is the diagonal matrix with diagonal entries
  $(d (1), \dotsc , d (n))$.
Typically, a random walk starts at a node $v$.
In this case, the distribution of the random walk at time $t$
  evolves according to $p_{t} = M^{t} \chi_{v}$.

We note that $\psi _{V}$ is the steady-state distribution of the
  random walk, and that $\psi _{S}$ is the restriction
  of that walk to the set $S$.

We will use the truncation operation defined by
\[
  \round{p} (u)
=
\begin{cases}
  p (u) & \text{if $p (u) \geq  d(u) \epsilon$,}\\
  0 & \text{otherwise}.
\end{cases}
\]

Our algorithm, \texttt{Nibble}, will generate the sequence of
  vectors starting at $\chi_{v}$ by the rules
\begin{align}
q_{t} = &
  \begin{cases}
\chi_{v} & \text{if $t = 0$, }
\\
 M r_{t-1}  & \text{otherwise,}
\end{cases}
\label{eqn:qt}
\\
r_{t} & = \round{q_{t}}.
\label{eqn:rt}
\end{align}

That is, at each time step, we will evolve the random walk one
  step from the current density, and then round every $q_{t} (u)$
  that is less than $d (u) \epsilon$ to 0.
Note that $q_{t}$ and $r_{t}$ are not necessarily probability vectors,
  as their components may sum to less than $1$.

In the statement of the algorithm and its analysis, we will use the
  following notation.
For a vector $p$, we let $\setj{j}{p}$ be the set of $j$ vertices $u$
  maximizing $p (u) / d (u)$, breaking ties lexicographically.
That is, $\setj{j}{p} = \setof{\pi (1), \dotsc , \pi (j)}$
  where $\pi$  is the permutation  such that
\[
p (\pi (i)) / d (\pi (i)) \geq p (\pi (i+1)) / d (\pi (i+1))
\]
  for all $i$,
  and $\pi (i) < \pi (i+1)$ when these two ratios are equal.
We then set
\[
  \lamj{j}{p} = \vol{S_{j} (p)} = \sum_{u \in \setj{j}{p}} d (u).
\]
Note that $\lamj{n}{p}$ always equals $2m$.

Following Lov\'asz and Simonovits \cite{LovaszSimonovitsFOCS}, we set
\begin{equation}\label{eqn:I}
  I (p, x) =
  \max_{\substack{w \in [0,1]^{n} \\
       \sum w (u) d (u) = x}}
  \sum_{u \in V} w (u) p (u) .
\end{equation}
This function $I (p,\cdot)$ is essentially the same as the function $h$
  defined by Lov\'asz and Simonovits---it only differs by a linear transformation.

We remark that for $x = \lamj{j}{p}$,
  $I (p,x) = p (\setj{j}{p})$, and that
  $I (p,x)$ is linear in $x$ between these points.
Finally, we let $I_{x} (p,x)$ denote the partial derivative of
  $I (p,x)$ with respect to $x$, with the convention that for
  $x = \lamj{j}{p}$,
\[
I_{x} (p,x) = \lim_{\delta\rightarrow 0}I_{x} (p,x - \delta ) =  p (\pi (j)) / d (\pi (j)),
\]
where $\pi$ is the permutation specified above
  so that $\pi (j) = \setj{j}{p} - \setj{j-1}{p}$.

As $p (\pi (i)) / d (\pi (i))$ is non-increasing,
  $I_{x} (p,x)$ is a non-increasing function in $x$ and
  $I (p,x)$ is a concave function in $x$.

During the course of our exposition, we will need to set many constants,
  which we collect here for convenience.
For each, we provide a suitable value and indicate where it is first used in the paper. 
\[
\begin{tabular}{l || l | l }
constant & value & where first used\\
\hline
$c_{1}$  & 200 & \eqref{eqn:t1}\\
$c_{2}$  & 280 & \eqref{eqn:f1phi}\\
$c_{3}$  & 1800 & \eqref{eqn:epsilon}\\
$c_{4}$  & 140 & \texttt{Nibble}, line C.4\\
$c_{5}$  & 20 & Definition~\ref{def:Sgb} \\
$c_{6}$  & 60 & Definition~\ref{def:Sgb}\\
\end{tabular}
\]
The following is an exhaustive list of the inequalities 
  we require these constants to satisfy.
\begin{align}
c_{2} & \geq 2 c_{4}  \label{eqn:c2geqc4}\\
c_{6} & \geq 2 c_{5} \label{eqn:c6geqc5}\\
c_{3} & \geq 8 c_{5} \label{eqn:c3geqc5}\\
c_{4} & \geq 4 c_{5} \label{eqn:c4geqc5}\\
\frac{1}{2 c_{6}} - \frac{1}{c_{3}} - \frac{1}{2 c_{5} c_{6}}& \geq \frac{1}{c_{4}}
\label{eqn:manyc1}
\\
\frac{1}{2 c_{5}} & \geq \frac{6}{5 c_{6}} + \frac{1}{c_{1}} \label{eqn:c5leqc1c6}\\
\frac{1}{5} & \geq 
\frac{1}{c_{5}} + \frac{4 c_{6}}{3 c_{3}} + \frac{1}{2 c_{1}} + \frac{1}{2 c_{2}}.
\label{eqn:manyc2}
\end{align}

Given a $\phi$, we set constants that will play a prominent role in our analysis:
\begin{align}
\ell  &\defeq  \ceiling{\log_{2}\left( \vol{V}/2 \right)} \label{eqn:l},\\
t_{1} & \defeq \ceiling{\frac{2}{\phi^{2}}
                 \ln \left(c_{1} (\ell +2) \sqrt{\vol{V}/2} \right)} ,
  \label{eqn:t1} \\
t_{h} & \defeq h t_{1}, \text{for $0 \leq h \leq \ell + 1$},
  \label{eqn:tj}\\
t_{last} & \defeq (\ell + 1) t_{1},  \text{and} \label{eqn:tmax}\\
f_{1} (\phi) & \defeq \frac{1}{c_{2} (\ell +2) t_{last}}. \label{eqn:f1phi}
\end{align}
Note that
\[
f_{1} (\phi) \geq \Omega \left(\frac{\phi^{2}}{\log^{3} \vol{V}} \right).
\]

\vskip 0.2in
\noindent
\fbox{
\begin{minipage}{6in}
\noindent $C = \mathtt{Nibble} (G, v, \phi,b)$\\
where $v$ is a vertex\\
$0 < \phi < 1$\\
$b$ is a positive integer.
\begin{enumerate}
\item  Set
\begin{equation}\label{eqn:epsilon}
\epsilon = 1/ (c_{3} (\ell+2) t_{last} 2^{b}).
\end{equation}

\item Set
  $q_{0} = \chi _{v}$ and
  $r_{0} = \round{q_{0}}$.

\item  For $t = 1$ to $t_{last}$
\begin{enumerate}
\item Set $q_{t} = M r_{t-1}$
\item  Set $r_{t} =  \round{ q_{t}}$.

\item  If there exists a $j $ such that
  \begin{enumerate}

\item [(C.1)] $\Phi (\setj{j}{q_{t}}) \leq \phi$,
\item [(C.2)]
    $\lamj{j}{q_{t}} \leq (5/6) \vol{V}$,
\item [(C.3)]
      $2^{b} \leq \lamj{j}{q_{t}}$, and
\item [(C.4)]
    $I_{x} (q_{t}, 2^{b}) \geq 1 / c_{4} (\ell +2) 2^{b}$.
  \end{enumerate}
  then return $C = \setj{j}{q_{t}}$
  and quit.
  \end{enumerate}
\item  Return $C = \emptyset$.
\end{enumerate}
\end{minipage}
}
\vskip 0.2in

Condition (C.1) guarantees that the set $C$ has low conductance.
Condition (C.2) ensures that it does not contain too much volume,
  while condition (C.3) ensures that it does not contain too little.
Condition (C.4) guarantees that many elements of $C$
  have large probability mass.
While it would be more natural to define condition (C.4) as a constraint
  on $I_{x} (q_{t}, \lamj{j}{q_{t}})$ instead of
  $I_{x} (q_{t}, 2^{b})$, our proof of correctness requires the latter.

In the rest of this section, we will
  prove the following theorem on the performance of \texttt{Nibble}.

\begin{theorem}[\texttt{Nibble}]\label{thm:Nibble}
\mbox{\rm \texttt{Nibble}}
  can be implemented so that
  on all inputs, it runs in time
  $  O (2^{b} (\log^{6} m ) / \phi^{4})$.
Moreover, \mbox{\rm \texttt{Nibble}} satisfies the following properties.
\begin{itemize}
\item [\mbox{\rm (N.1)}]
When
$C =\mbox{\rm \texttt{Nibble}} (G, v,\phi, b)$ is non-empty,
\[
\conduc{}{C} \leq \phi \quad  \mbox{\rm and}\quad  \vol{C} \leq (5/6)\vol{V}.
\]
\item [\mbox{\rm (N.2)}]
  Each set $S$ satisfying
\[
\vol{S} \leq (2/3) \vol{V} \quad \mbox{\rm and} \quad   \conduc{}{S} \leq f_{1} (\phi )
\]
has a subset $S^{g}$ such that
\begin{itemize}
\item [\mbox{\rm (N.2.a)}] $\vol{S^{g}} \geq \vol{S}/2$, and

\item [\mbox{\rm (N.2.b)}]  $v \in S^{g}$ and
  $C =\mbox{\rm \texttt{Nibble}} (G, v,\phi, b) \not = \emptyset$ imply
  $\vol{C \cap S} \geq 2^{b-1}$.
\end{itemize}
\item [\mbox{\rm (N.3)}] The set $S^{g}$ may be partitioned into
  subsets $S^{g}_{0}, \dotsc , S^{g}_{\ell}$
  such that if $v \in S^{g}_{b}$, then the
  set $C$ output by $\mathtt{Nibble} (G, v, \phi,b)$
  will not be empty.
\end{itemize}
\end{theorem}


\subsection{Basic Inequalities about Random Walks}

We first establish some basic inequalities
  that will be useful in our analysis.
Readers who are eager to see the analysis of \texttt{Nibble}
  can first skip this subsection.
Suppose $G = (V,E)$ is an undirected graph.
Recall $M = (A D^{-1} + I) / 2$,  where $A$ is the
  adjacency matrix of $G$.

\begin{proposition}[Monotonicity of Mult by $M$]\label{pro:infnorm}
For all non-negative vectors $p$,
\[
  \infnorm{D^{-1} (M p)}
 \leq
  \infnorm{D^{-1} p}.
\]
\end{proposition}
\begin{proof}
Applying the transformation $z = D^{-1}p$, we see that it
  is equivalent to show that for all $z$
\[
  \norm{D^{-1} M D z}_{\infty } \leq \norm{z}_{\infty }.
\]
To prove this, we note that
  $D^{-1} M D = D^{-1} (AD^{-1}+I) D/2 = M^{T}$,
  and the sum of the entries in each row of this matrix is 1.
\end{proof}

\begin{definition}\label{def:DS}
For a set $S \subseteq V$, we define the matrix $D_{S}$ to be
  the diagonal matrix such that
  $D_{S} (u,u) = 1$ if $u \in S$ and $0$ otherwise.
\end{definition}

\begin{proposition}\label{pro:DS}
For every $S \subseteq V$, all
 non-negative vectors $p$ and $q$,
  and every $t \geq 1$,
\[
  p^{T} (D_{S} M)^{t} q \leq p^{T} M^{t} q.
\]
\end{proposition}
\begin{proof}
For $t=1$, we observe
\[
  p^{T} (M) q
=
  p^{T} ((D_{S} + D_{\bar{S}}) M) q
=
  p^{T} (D_{S} M) q +
  p^{T} (D_{\bar{S}} M) q
\geq
  p^{T} (D_{S} M) q,
\]
as $p$, $q$,  $D_{\bar{S}}$, and $M$ are all non-negative.
The proposition now follows by induction.
\end{proof}

\begin{proposition}[Escaping Mass]\label{pro:isopLeaving}
For all $t \geq 0$ and for all $S\subset V$,
\[
\form{\bvec{1}}{(D_{S} M)^{t} \psi _{S}}
\geq
1 - t \Phi_{V} (S)/2.
\]
\end{proposition}
\begin{proof}
Note that $M \psi _{S}$ is the distribution after
  a single-step walk from a random vertex in $S$
  and $\bvec{1}^{T}D_{S} (M\psi _{S})$
  is the probability that the walk stays inside $S$.
Thus, $\form{\bvec{1}}{(D_{S} M)^{t} \psi _{S}}$
  is the probability that a $t$-step walk starting
  from a random vertex in $S$ stays entirely in $S$.

We first prove by induction that for all $t \geq 0$,
\begin{equation}\label{eqn:escapMass1}
\infnorm{D^{-1} (D_{S} M)^{t} \psi _{S}} \leq 1/\vol{S}.
\end{equation}
The base case, $t= 0$, follows from the fact that
  $\infnorm{D^{-1} \psi _{S}} = 1/\vol{S}$.
To complete the induction, observe that if
  $x$ is a non-negative vector such that
  $\infnorm{D^{-1} x} \leq 1/\vol{S}$, then
\[
\infnorm{D^{-1} (D_{S} M) x}
=
\infnorm{D_{S} D^{-1}  M x}
\leq
\infnorm{D^{-1}  M x}
\leq
\infnorm{D^{-1}  x}
\leq
1/\vol{S},
\]
where the second-to-last inequality follows from
  Proposition~\ref{pro:infnorm}.

We will now prove that for all $t$,
\[
\form{\bvec{1}}{(D_{S} M)^{t} \psi _{S}}
-
\form{\bvec{1}}{(D_{S} M)^{t+1} \psi _{S}}
\leq  \Phi_{V} (S)/2,
\]
from which the proposition follows, as $\bvec{1}^{T} \psi_{S} = 1$.

Observing that $\bvec{1}^{T} M = \bvec{1}^{T}$,
  we compute
\begin{align*}
\lefteqn{\form{\bvec{1}}{(D_{S} M)^{t} \psi _{S}}
-
\form{\bvec{1}}{(D_{S} M)^{t+1} \psi _{S}}}\\
& =
\bvec{1}^{T} (I - D_{S} M) (D_{S} M)^{t} \psi _{S}\\
& =
\bvec{1}^{T} (M - D_{S} M) (D_{S} M)^{t} \psi _{S}\\
& =
\bvec{1}^{T} (I - D_{S} ) M (D_{S} M)^{t} \psi _{S}\\
& =
\chi_{\bar{S}}^{T} M (D_{S} M)^{t} \psi _{S}\\
& =
(1/2) \chi_{\bar{S}}^{T} (I + A D^{-1}) (D_{S} M)^{t} \psi _{S}\\
& =
(1/2) \chi_{\bar{S}}^{T} (A D^{-1}) (D_{S} M)^{t} \psi _{S}
\quad \quad \quad \quad \quad \quad \text{ (as $\chi_{\bar{S}}^{T} I D_{S} = \bvec{0}$)}\\
& \leq
(1/2) \sizeof{E (S,V-S)} \infnorm{D^{-1} (D_{S} M)^{t} \psi _{S}}\\
& \leq
\frac{1}{2}\frac{\sizeof{E (S,V-S)}}{\vol{S}}
\quad \quad \quad  \quad \quad \quad  \quad \quad \text{(by inequality~\eqref{eqn:escapMass1})}\\
& \leq
 \Phi_{V} (S)/2.
\end{align*}
\end{proof}

\subsection{The Analysis of \texttt{Nibble}}

Our analysis of \texttt{Nibble} consists of three main steps.
First, we define the sets $S^{g}$ mentioned in
  Theorem \ref{thm:Nibble} and establish property (N.2).
We then refine the structure of $S^{g}$ to
  define sets $S_{b}^{g}$ and prove property (N.3).
The sets $S^{g}$ and $S^{g}_{b}$ are defined in terms
  of the distributions of
  random walks from a vertex in $S$, without reference to the
  truncation we perform in the algorithm.
We then analyze the impact of truncation used in \texttt{Nibble}
  and extend the theory of Lov\'asz and Simonovits \cite{LovaszSimonovits}
  to truncated random walks.

\subsection*{Step 1: $S^{g}$ and its properties}


\begin{definition}[$S^{g}$]\label{def:sg}
For each set $S \subseteq V$, we define
  $S^{g}$ to be the set of nodes $v$ in $S$ such that
  for all $t \leq t_{last}$,
\[
  \form{\chi _{\bar{S}}}{M^{t} \chi _{v}}
\leq
  t_{last} \Phi (S).
\]
\end{definition}

Note that $\form{\chi _{\bar{S}}}{M^{t} \chi _{v}}$
  denotes the probability that a $t$-step random walk
  starting from $v$ terminates outside $S$.
Roughly speaking, $S^{g}$ is the set of vertices $v \in S$
  such that a random walk from $v$
  it is reasonably likely to still be in $S$ after $t_{last}$
  time steps.
We will prove the following bound on the volume of $S^{g}$.
\begin{lemma}[Volume of $S^{g}$]\label{lem:sizeSg}
\[
  \vol{S^{g}} \geq \vol{S}/2.
\]
\end{lemma}
\begin{proof}
Let $S \subseteq V$, and let $D_{S}$ be the diagonal matrix such that
  $D_{S} (u,u) = 1$ if $u \in S$ and $0$ otherwise.
For $t \geq 0$,
\begin{align*}
\form{\chi _{\bar{S}}}{M^{t} \chi _{v}}
& =
\left(\bvec{1} - \chi_{S} \right)^{T} M^{t} \chi_{v}\\
& =
\bvec{1}^{T}  \chi_{v} - \chi_{S}^{T} M^{t}  \chi_{v}\\
& =
1 - \form{\chi _{S}}{M^{t} \chi _{v}}\\
& \leq
1 - \form{\bvec{1}}{(D_{S} M)^{t} \chi _{v}},
\quad \text{by Proposition~\ref{pro:DS},}
\\
& \leq
1 - \form{\bvec{1}}{(D_{S} M)^{t_{last}} \chi _{v}},
\end{align*}
as $\form{\bvec{1}}{(D_{S} M)^{t} \chi _{v}}$
  is a non-increasing function of $t$.
Define
\[
 S' = \setof{v : 1 - \form{\bvec{1}}{(D_{S} M)^{t_{last}} \chi _{v}}
  \leq
   t_{last} \Phi (S)}.
\]
So, $S' \subseteq S^{g}$, and it
  suffices to prove that
  $\vol{S'} \geq \vol{S}/ 2$.

Applying Proposition~\ref{pro:isopLeaving}, we obtain
\begin{align*}
t_{last} \Phi (S) / 2
& \geq
 1 - \form{\bvec{1}}{(D_{S} M)^{t_{last}} \psi _{S}}
\\
 & =
  \sum _{v \in S}
  \frac{d (v)}{\vol{S}}
\left( 1 - \form{\bvec{1}}{(D_{S} M)^{t_{last}} \chi _{v}} \right)
\\
 & >
  \sum _{v \in S - S'}
  \frac{d (v)}{\vol{S}}
   t_{last} \Phi (S), \quad \text{by the definition of $S'$}\\
 & = \frac{\vol{S - S'}}{\vol{S}}
   t_{last} \Phi (S).
\end{align*}
So, we may conclude
\[
\frac{\vol{S - S'}}{\vol{S}}
<
\frac{1}{2},
\]
from which the lemma follows.
\end{proof}

We now prove the following lemma,  which says  that
  if \texttt{Nibble} is started from any
  $v \in S^{g}$ with parameter $b$ and returns a non-empty set $C$,
  then  $\vol{C \cap S} \geq 2^{b-1}$.

\begin{lemma}[N2]\label{lem:N3}
Let $S \subseteq V$ be a set of vertices such that
  $\Phi (S) \leq f_{1} (\phi )$.
If \texttt{Nibble} is run with parameter $b$,
  is started at a $v \in S^{g}$,
  and outputs a non-empty set $C$, then
  $\vol{C \intersect S} \geq 2^{b-1}$.
\end{lemma}
\begin{proof}
For $v \in S^{g}$, let $q_{t}$ be given by \eqref{eqn:qt} and \eqref{eqn:rt}.
Then, for $t \leq t_{last}$,
\[
  \chi_{\bar{S}}^{T} q_{t}
\leq \chi_{\bar{S}}^{T} M^{t} \chi_{v}
\leq  \Phi (S) t_{last}
  \leq  f_{1} (\phi ) t_{last}
  \leq \frac{1}{c_{2} (\ell +2) },
\]
where the second inequality follows from the definition of $S^{g}$.

Let $t$ be the index of the step at which the set $C$ is generated.
Let $j'$ be the least integer such that
  $\lamj{j'}{q_{t}} \geq 2^{b}$.
Condition (C.3) implies $j' \leq j$.
As $I_{x}$ is non-increasing in its second argument and
  constant between $2^{b}$ and  $\lamj{j'}{q_{t}}$,
  Condition $(C.4)$ guarantees that for all
  $u \in \setj{j'}{q_{t}}$,
\[
q_{t} (u) / d (u) \geq 1/ c_{4} (\ell +2) 2^{b}.
\]
Thus,
\begin{multline*}
\vol{\setj{j'}{q_{t}} \intersect \bar{S}}
=
  \sum_{u \in \setj{j'}{q_{t}} \intersect \bar{S}} d (u)
\leq
  \sum_{u \in \setj{j'}{q_{t}} \intersect \bar{S}} c_{4} (\ell +2) 2^{b} q_{t} (u)
\\
\leq
  c_{4} (\ell +2) 2^{b} (\chi_{\bar{S}}^{T} q_{t})
\leq 
  \frac{c_{4} (\ell +2) 2^{b}}{c_{2} (\ell +2)}
\leq
  2^{b-1},
\end{multline*}
by \eqref{eqn:c2geqc4}.
So, $\vol{\setj{j'}{q_{t}} \intersect S} \geq 2^{b-1}$,
 and, as $j' \leq j$,
\[
\vol{\setj{j}{q_{t}}\intersect S} 
\geq
\vol{\setj{j'}{q_{t}} \intersect S} \geq 2^{b-1}.
\]
\end{proof}

\subsection*{Step 2: Refining $S^{g}$}

Before defining the sets $S^{g}_{b}$,
  we first recall some of the facts we can infer about
  the function $I$ from the work of Lov\'asz and Simonovits.
These facts will motivate our definitions and analysis.

In the first part of the proof of Lemma~1.4 of~\cite{LovaszSimonovitsFOCS},
 Lov\'asz and Simonovits prove
\begin{lemma}\label{lem:LSeasy}
For every non-negative vector $p$ and every $x$,
\begin{equation}\label{eqn:LSsimple}
  I (M p, x) \leq I (p, x).
\end{equation}
\end{lemma}

For each $p_{t}$, $I (p_{t},x)$ is a concave function that
  starts at $(0,0)$ and goes to $(\vol{V},1)$.
Lemma~\ref{lem:LSeasy} says that for each $t$, the curve defined
  by $I (p_{t+1}, \cdot)$ lies below the curve defined
  by $I (p_{t}, \cdot)$.
In particular,
\begin{equation}\label{eqn:below}
\forall x, I (p_{t+1}, x) \leq I (p_{t},x).
\end{equation}

If none of the sets $\setj{j}{p_{t+1}}$ has conductance less than $\phi$,
  then Lov\'asz and Simonovits prove a bound on how far below
  $I (p_{t}, \cdot)$ the curve of
  $I (p_{t+1}, \cdot)$ must lie.
The following Lemma is a special case of Lemma~1.4 of~\cite{LovaszSimonovits},
  restricted to points $x$ of the form $\lamj{j}{M p}$.
Lov\'asz and Simonovits~\cite{LovaszSimonovitsFOCS} claim that the following
  is true for all $x$, but point out in the journal version of their 
  paper~\cite{LovaszSimonovits} that this claim was false.
Fortunately, we do not need the stronger claim.

\begin{lemma}\label{lem:LShard}
For any non-negative vector $p$,
  if $\Phi (\setj{j}{M p}) \geq \phi$, then
 for $x = \lamj{j}{M p}$,
\[
  I (M p, x)
\leq
\frac{1}{2}
\left(
I \big(p, x - 2 \phi \widehat{x}\big)
 +
I \big(p, x + 2 \phi \widehat{x} \big)
 \right),
\]
where $\widehat{x}$ denotes $\min (x, 2m-x)$.
\end{lemma}

The mistake in~\cite{LovaszSimonovitsFOCS} is the assertion in the beginning
  of the proof that the inequality holds for all $x$ if it holds for all
  $x$ of form $\lamj{j}{M p}$.

When this lemma applies, one may draw a chord across the curve
  of $I (p_{t}, \cdot)$ around $x$ of width proportional to
  $\phi$, and know that $I (p_{t+1},x)$ lies below.
Thus, we know that if none of the sets $\setj{j}{p_{t}}$
  has conductance less than $\phi$, then the curve
  $I (p_{t}, \cdot)$ will approach a straight line.
On the other hand, Proposition~\ref{pro:isopLeaving} will tell
  us that some point of $I (p_{t_{last}}, \cdot)$ lies well above this line
  (see Lemma~\ref{lem:lowerBound}).

We will now define the sets
 $S^{g}_{b}$
  for $b = 1, \dotsc , \ell $,
  simultaneously with two quantities---$h_{v}$ and $x_{h}$,
  where $h_{v}$ is such that \texttt{Nibble} will stop
  between iterations $t_{h_{v}-1}$ and $t_{h_{v}}$ and $x_{h_{v}}$
  is the $x$-coordinate of a point that 
  may be shown in Lemmas~\ref{lem:lowerBound} and \ref{lem:C123} to
  contradict the conclusion
  of Lemma~\ref{lem:LShard} and thereby enable us to find a set
  of low conductance.

\begin{definition}[$x_{h}$, $h_{v}$ and $S^{g}_{b}$]\label{def:Sgb}
Given a $v \in S^{g}$, let $p_{t} = M^{t} \chi_{v}$.
For $0 \leq h \leq \ell+1$,
  define $x_{h} (v)$ to be the real number such that
\[
  I (p_{t_{h}}, x_{h} (v)) = \frac{h+1}{c_{5} (\ell +2)}.
\]
We write $x_{h}$ instead of $x_{h} (v)$ when $v$ is clear from context.
Define
\[
h_{v} =
\begin{cases}
\ell+1  &  \text{if $x_{\ell} (v) \geq 2m / c_{6}(\ell +2)$, and}
\\
\min \setof{h : x_{h} \leq 2 x_{h-1}}  & \text{otherwise}.
\end{cases}
\]
We define
\[
  S^{g}_{0} = \setof{v : x_{h_{v}-1} (v) < 2},
\]
and for $b = 1, \dotsc , \ell$, we define
\[
  S^{g}_{b} = \setof{v : x_{h_{v}-1} (v) \in [2^{b}, 2^{b+1})}
\]
\end{definition}

\begin{proposition}\label{pro:Sgb}
The quantities $h_{v}$ are well-defined and
  the sets $S^{g}_{b}$ partition $S^{g}$.
Moreover, $x_{h-1} < x_{h}$ for all $h$.
\end{proposition}
\begin{proof}
It follows from the definition of
  $I$ that 
  for a probability vector $p$ the slope of $I (p, \cdot)$
  is always less than 1, and so
  $x_{0} \geq 1 / c_{5} (\ell +2)$.
If $x_{\ell} < \vol{V} / c_{6} (\ell +2)$,
  then 
\[
x_{\ell}/ x_{0} <
\frac{\vol{V} c_{5} (\ell + 2) }{c_{6} (\ell + 2)}
\leq 
 \vol{V}/2, \quad \text{(by inequality \eqref{eqn:c6geqc5})}
\]
  so there is an integer $h \leq \ell$ such that
  $x_{h} \leq 2 x_{h-1}$, and so the quantities $h_{v}$ are well-defined.

To see that the sets $S^{g}_{b}$ partition $S^{g}$, it now suffices to
  observe that $x_{h_{v}-1} < \vol{V} \leq 2^{\ell + 1}$.

Finally, to show that $x_{h-1} < x_{h}$,
  we apply Lemma~\ref{lem:LSeasy} to show
\[
  I (p_{t_{h}}, x_{h-1}) \leq 
  I (p_{t_{h-1}}, x_{h-1})
 = \frac{h}{c_{5} (\ell + 2)}.
\]
As $I (p_{t_{h}}, \cdot)$ is non-decreasing
  and
\[
I (p_{t_{h}}, x_{h}) > \frac{h}{c_{5} (\ell + 2)},
\]
we can conclude that $x_{h} > x_{h-1}$.
\end{proof}

\subsection*{Step 3: Clustering and truncated random walks}

We now establish that vectors produced by the
  truncated random walk do not differ too much from those
  produced by the standard random walk.
\begin{lemma}[Low-impact Truncation]\label{lem:rounding}
For all $u \in V$ and $t$,
\begin{equation}\label{eqn:roundingp}
  p_{t} (u) \geq q_{t} (u) \geq   r_{t} (u) \geq p_{t} (u) -  t \epsilon d (u).
\end{equation}
For all $t$ and $x$,
\begin{equation}\label{eqn:roundingI}
I (p_{t},x) \geq I (q_{t}, x) \geq  I (r_{t}, x) \geq I (p_{t}, x) -  \epsilon x t.
\end{equation}
\end{lemma}
\begin{proof}
The left-hand inequalities of \eqref{eqn:roundingp} are trivial.
To prove the right-hand inequality of \eqref{eqn:roundingp}, we consider
  $p_{t} - \round{p_{t}}$, observe that by definition
\[
  \norm{D^{-1} \left(p_{t} - \round{p_{t}} \right) }_{\infty }
  \leq \epsilon,
\]
and then apply Proposition~\ref{pro:infnorm}.
Inequality \eqref{eqn:roundingI} then
  follows from \eqref{eqn:I}.
\end{proof}

\begin{lemma}[Lower bound on I]\label{lem:lowerBound}
Let $S \subseteq V$ be a set of vertices
  such that
  $\vol{S} \leq (2/3) \vol{V}$ and
  $\Phi (S) \leq f_{1} (\phi )$, and
  let $v$ lie in $S^{g}_{b}$.
Define $q_{t}$ by running 
 \texttt{Nibble} is with parameter $b$.
\begin{itemize}
\item [1.] If $x_{\ell} (v) \geq  2m / c_{6} (\ell + 2)$, then
\begin{equation}\label{eqn:lower1}
I (q_{t_{\ell + 1}}, (2/3) (2m)) \geq 
1 - \frac{1}{c_{2} (\ell + 2)} - \frac{4 c_{6}}{3 c_{3}}
\end{equation}
\item [2.] Otherwise,
\begin{equation}\label{eqn:lower2}
I (q_{t_{h_{v}}}, x_{h_{v}}) \geq \frac{h_{v} + 1/2}{c_{5} (\ell +2)}
\end{equation}
\end{itemize}
\end{lemma}
\begin{proof}
In the case $x_{\ell} (v) \geq 2m / c_{6} (\ell + 2)$, we compute
\begin{align*}
  I (p_{t_{\ell + 1}}, (2/3) (2m)) 
& \geq 
  I (p_{t_{\ell + 1}}, \vol{S}) 
\\
& \geq \sum_{u \in S} p_{t_{\ell + 1}} (u) 
& \text{by \eqref{eqn:I}}
\\
& =
  \chi_{S}^{T} p_{t_{\ell + 1}}
\\
& \geq 
  1- t_{last} f_{1} (\phi),
& \text{by the definition of $S^{g}$}
\\
& \geq 
1 - \frac{1}{c_{2} (\ell + 2)}.
& \text{by \eqref{eqn:f1phi}}.
\end{align*}
As $2^{b+1} > x_{\ell} \geq 2m / c_{6} (\ell + 2)$, we may use
 Lemma~\ref{lem:rounding} 
  to show
\[
I (q_{t_{\ell + 1}}, (2/3) (2m)) \geq 
1 - \frac{1}{c_{2} (\ell + 2)} - \epsilon (4m/3) t_{last}
=
1 - \frac{1}{c_{2} (\ell + 2)} - \frac{4 c_{6}}{3 c_{3}},
\]
by  \eqref{eqn:epsilon}.

If $x_{\ell} (v) < 2m / c_{6} (\ell + 2)$, we compute
\begin{align}
I (q_{t_{h_{v}}}, x_{h_{v}})
& \geq
I (p_{t_{h_{v}}}, x_{h_{v}}) - \epsilon t_{last} x_{h_{v}}
\quad \text{(by Lemma~\ref{lem:rounding})}
\notag
\\
& =
\frac{h_{v} + 1}{c_{5} (\ell +2)} - \epsilon t_{last} x_{h_{v}}
\notag
\\
& =
\frac{h_{v} + 1}{c_{5} (\ell +2)} -
  \frac{x_{h_{v}}}{c_{3} (\ell+2)  2^{b}}
\quad \text{(by \eqref{eqn:epsilon})}
\notag
\\
& \geq
\frac{h_{v} + 1}{c_{5} (\ell +2)} -
  \frac{2 x_{h_{v}-1}}{c_{3} (\ell+2) 2^{b}}
\quad \text{(as $x_{h_{v}} \leq 2 x_{h_{v}-1}$)}
\notag
\\
& >
\frac{h_{v} + 1}{c_{5} (\ell +2)} -
  \frac{2^{b+2}}{c_{3} (\ell+2) 2^{b}}
\quad \text{(as $x_{h_{v}-1} < 2^{b+1}$)}
\notag
\\
&\geq
\frac{h_{v} + 1/2}{c_{5} (\ell +2)},
\quad \text{by \eqref{eqn:c3geqc5}}.
\notag
\end{align}
\end{proof}

\begin{lemma}[C.4]\label{lem:C4}
Let $S \subseteq V$ be a set of vertices
  such that
  $\vol{S} \leq (2/3) \vol{V}$ and
  $\Phi (S) \leq f_{1} (\phi )$, and
  let $v$ lie in $S^{g}_{b}$.
If \texttt{Nibble} is run with parameter $b$,
  then for all $t \in (t_{h_{v}-1}, t_{h_{v}}]$,
  condition (C.4) is satisfied.
\end{lemma}
\begin{proof}
We first consider the case in which $x_{\ell } < 2m / c_{6} (\ell +2)$,
  which by definition implies $x_{h_{v}} \leq  2 x_{h_{v}-1}$.

In this case,
  we have
\[
I (q_{t}, x_{h_{v}-1}) \leq I (p_{t}, x_{h_{v}-1})
\leq I (p_{t_{h_{v}-1}}, x_{h_{v}-1})
  = h_{v} / c_{5} (\ell +2),
\]
where the first inequality follows from Lemma~\ref{lem:rounding}
and the second follows from Lemma~\ref{lem:LSeasy}.

As $I_{x} (q_{t},x)$ is non-increasing in $x$
  and $x_{h_{v}-1} < x_{h_{v}} \leq 2 x_{h_{v}-1}$, we have
\[
 I_{x} (q_{t}, x_{h_{v}-1})
\geq
\frac{I (q_{t}, x_{h_{v}}) - I (q_{t}, x_{h_{v}-1})
}{x_{h_{v}} - x_{h_{v}-1}}
\geq
\frac{I (q_{t_{h_{v}}}, x_{h_{v}}) - I (q_{t}, x_{h_{v}-1})
}{x_{h_{v}} - x_{h_{v}-1}}
\geq
\frac{1 / 2 }{c_{5} (\ell +2) x_{h_{v}-1}},
\]
where the second inequality follows from Lemma~\ref{lem:LSeasy} and the definition
  of $q_{t}$, and the last follows from \eqref{eqn:lower2}.

If $x_{h_{v}-1} \geq 2$, then $b \geq 1$ 
  and we have $2^{b} \leq x_{h_{v}-1} < 2^{b+1}$,
 and so
\[
 I_{x} (q_{t}, 2^{b})
 \geq
 I_{x} (q_{t}, x_{h_{v}-1})
  \geq
 \frac{1 }{ 2 c_{5} (\ell +2) 2^{b+1}},
\]
and by \eqref{eqn:c4geqc5} condition (C.4) is satisfied.

If $x_{h_{v}-1} < 2$, then $b = 0$, and so 
  $I_{x} (q_{t}, 2^{b}) = I_{x} (q_{t},x)$ for all $x < 1$,
  which implies
\[
  I_{x} (q_{t}, 2^{b}) \geq 
  I_{x} (q_{t}, x_{h_{v}-1}) \geq
  \frac{1  }{2 c_{5} (\ell +2) x_{h_{v}-1}} >
 \frac{1 }{2 c_{5} (\ell +2) 2^{b}}
\]
and condition (C.4) is satisfied.

If $x_{\ell} \geq 2m  / c_{6} (\ell +2)$, in which case
  $h_{v} = \ell + 1$
  and $x_{\ell} < 2^{b+1}$,
  we apply Lemma~\ref{lem:rounding}  to show that
  for all $t \in (t_{\ell },   t_{\ell +1}]$,
\[
I (q_{t}, 2m)
\geq
I (p_{t}, 2m) - 2m \epsilon t_{last}
=
1 - \frac{2m}{c_{3} (\ell+2) 2^{b}}
\geq
1 - \frac{2m}{c_{3} (\ell+2) x_{\ell}/2}
\geq
1 - \frac{2 c_{6}}{c_{3}}.
\]
On the other hand,
\[
  I (q_{t}, x_{\ell })
\leq
  I (p_{t}, x_{\ell })
\leq
  I (p_{t_{\ell}}, x_{\ell })
<  1/c_{5}.
\]
As $I_{x} (q_{t}, \cdot)$ is non-decreasing and $x_{\ell} \geq 2^{b}$,
  we have
\begin{multline*}
I_{x} (q_{t}, 2^{b})
\geq 
I_{x} (q_{t}, x_{\ell})
\geq
\frac{
I (q_{t}, 2m)  -  I (q_{t}, x_{\ell })
}{
2m - x_{\ell }
}\\
\geq
\frac{1}{2m}
\left(
 1 - \frac{2 c_{6}}{c_{3}} - \frac{1}{c_{5}}
 \right)
\geq
\frac{1}{c_{6} (\ell + 2) 2^{b+1}}
\left(
 1 - \frac{2 c_{6}}{c_{3}} - \frac{1}{c_{5}}
 \right),
\end{multline*}
as $2^{b+1} > x_{\ell} \geq 2m / c_{6} (\ell+2)$,
and so by \eqref{eqn:manyc1} condition (C.4) is satisfied.
\end{proof}

It remains to show that conditions (C.1--3) are met for some
  $t \in (t_{h_{v}-1}, t_{h_{v}}]$.
We will do this by showing that if at least one of these
  conditions fail for every $j$ and every
  $t \in (t_{h_{v}-1}, t_{h_{v}}]$,
  then the curve $I (q_{t_{h}}, \cdot)$ will be too low, in violation
  of Lemma~\ref{lem:lowerBound}.

\begin{lemma}\label{lem:C1help}
If there exists a $\beta > 0$ and an
 $h \in [1,\ell + 1]$, 
  such that for all
  $t \in (t_{h-1}, t_{h}]$ and for all $j$
  either
\begin{enumerate}
\item [1.]  $\Phi (\setj{j}{q_{t}}) \geq \phi$,
\item [2.] $\lamj{j}{q_{t}} > (5/6) 2m$, or
\item [3.] $I (q_{t}, \lamj{j}{q_{t}}) < \beta $, 
\end{enumerate}
then, for all $x$
\[
  I (q_{t_{h}}, x)
<
  \beta + \frac{3 x}{5 m}
+
  \sqrt{\widehat{x}} \left(1 - \frac{\phi^{2}}{2} \right)^{t_{1}}.
\]
\end{lemma}
\begin{proof}
We will prove by induction that the conditions of the lemma
  imply that for all $t \in [t_{h-1}, t_{h}]$ and
  all $x$,
\begin{equation}\label{eqn:It}
  I (q_{t}, x)
<
  \beta + \frac{3 x}{5 m}
+
  \sqrt{\widehat{x}} \left(1 - \frac{\phi^{2}}{2} \right)^{t - t_{h-1}}.
\end{equation}
The base case is when $t = t_{h-1}$, in which case
  \eqref{eqn:It} is satisfied because
\begin{itemize}
\item For $1 \leq x \leq 2m-1$,
  $I (q_{t}, x) \leq I (q_{t}, 2m) \leq 1 \leq \sqrt{\widehat{x}}$.
\item For $0 \leq x \leq 1$,
  we have $I (q_{t}, x) \leq \sqrt{\widehat{x}}$ as both are $0$
  at $x = 0$, the right-hand term dominates at $x= 1$, the
  left-hand term is linear in this region, and the right-hand term is concave.

\item For $2m-1 \leq x \leq 2m$,
  we note that at $x=2m$,
  $I (q_{t},x) = 1 < 3x/5m$, and that
  we already know the right-hand term dominates  at $x=2m-1$.
  The inequality then follows from the facts that
  left-hand term is linear in this region, and the right-hand term is concave.
\end{itemize}

Let
\[
f (x) \defeq \sqrt{\widehat{x}}.
\]
Lov\'asz and Simonovits~\cite{LovaszSimonovitsFOCS} observe that
\begin{equation}\label{eqn:LS}
\frac{1}{2}
\left(
f (x - 2 \phi \widehat{x}) +
f (x + 2 \phi \widehat{x})
 \right)
\leq
f (x) \left(1 - \frac{\phi^{2}}{2} \right).
\end{equation}
We now prove that \eqref{eqn:It} holds for $t$, assuming it holds for $t-1$,
  by considering three cases.
As the right-hand side is concave and the left-hand side is piecewise-linear
  between points of the form $\lamj{j}{q_{t}}$,
  it suffices to prove the inequality at the points
  $\lamj{j}{q_{t}}$.
If $x = \lamj{j}{q_{t}}$ and 
  $I (q_{t}, x) \leq \beta$, then \eqref{eqn:It} holds trivially.
Similarly, if $x = \lamj{j}{q_{t}} > (5/6) 2m $,
  then \eqref{eqn:It} holds trivially as well, as
  the left-hand side is at most 1, and the right hand side is at least 1.
In the other cases, we have $\Phi (\setj{j}{q_{t}}) \geq \phi$,
  in which case we may apply Lemma~\ref{lem:LShard} to show
  that for $x = \lamj{j}{q_{t}}$
\begin{align*}
  I (q_{t}, x)
& = 
  I (M r_{t-1}, x) & \text{by definition}
\\
& \leq 
\frac{1}{2}
\left(
I \big(r_{t-1}, x - 2 \phi \widehat{x}\big)
 +
I \big(r_{t-1}, x + 2 \phi \widehat{x} \big)
 \right) 
& \text{by Lemma~\ref{lem:LShard}}
\\
& \leq 
\frac{1}{2}
\left(
I \big(q_{t-1}, x - 2 \phi \widehat{x}\big)
 +
I \big(q_{t-1}, x + 2 \phi \widehat{x} \big)
 \right) 
\\
& <
\frac{1}{2}
\left[
\beta + \frac{3 (x - 2 \phi \widehat{x})}{5m}
+
\sqrt{x - 2 \phi \widehat{x}} \left(1 - \frac{\phi^{2}}{2} \right)^{t-1-t_{h-1}}
\right.
\\
& \quad \left. + 
\beta + \frac{3 (x + 2 \phi \widehat{x})}{5m}
+
\sqrt{x + 2 \phi \widehat{x}} \left(1 - \frac{\phi^{2}}{2} \right)^{t-1-t_{h-1}}
\right]
& \text{by induction}
\\
& =
\beta + \frac{3 x}{5 m}
+
\frac{1}{2}
\left(
\sqrt{x - 2 \phi \widehat{x}}
+
\sqrt{x + 2 \phi \widehat{x}}
\right)
 \left(1 - \frac{\phi^{2}}{2} \right)^{t-1-t_{h-1}}
\\
& \leq 
\beta + \frac{3 x}{5 m}
+
\sqrt{\widehat{x}}
 \left(1 - \frac{\phi^{2}}{2} \right)^{t-t_{h-1}},
\end{align*}
by~\eqref{eqn:LS}.
\end{proof}

We now observe that $t_{1}$ has been chosen to ensure
\begin{equation}\label{eqn:tEnsure}
  \sqrt{\widehat{x}} \left(1 - \frac{\phi^{2}}{2} \right)^{t_{1}}
<
  \frac{1}{c_{1} (\ell +2)}.
\end{equation}

\begin{lemma}[C.1-3]\label{lem:C123}
Let $S$ be a set of vertices such that
 $\vol{S} \leq (2/3) (2m)$ and
  $\Phi (S) \leq  f_{1} (\phi )$, and
  let $v$ lie in $S^{g}_{b}$.
If \texttt{Nibble} is run with parameter $b$,
  then there exists a $t \in ( t_{h_{v}-1}, t_{h_{v}}]$ and a $j$
  for which conditions (C.1-3) are satisfied.
\end{lemma}
\begin{proof}
We first show that for $t \in ( t_{h_{v}-1}, t_{h_{v}}]$, (C.3) is implied
  by $I (q_{t}, \lambda_{j} (q_{t})) \geq h_{v}/c_{5} (\ell + 2)$.
To see  this, note that for $b > 0$
\begin{align*}
I (q_{t}, 2^{b}) 
& \leq 
I (q_{t}, x_{h_{v}-1}),
& \text{as $x_{h_{v}-1} \geq 2^{b}$,}
\\
& \leq 
I (p_{t}, x_{h_{v}-1}),
& \text{by Lemma~\ref{lem:rounding},}
\\
& \leq 
I (p_{t_{h_{v}-1}}, x_{h_{v}-1})
& \text{by Lemma~\ref{lem:LSeasy},}
\\
& =
\frac{h_{v}}{c_{5} (\ell + 2)}.
\end{align*}
So, $I (q_{t}, \lambda_{j} (q_{t})) \geq h_{v}/ c_{5} (\ell + 2)$
  implies $\lambda_{j} (q_{t}) \geq 2^{b}$,
 and we may prove the lemma by exhibiting a $t$ and $j$ for which
  (C.1), (C.2) and $I (q_{t}, \lambda_{j} (q_{t})) \geq h_{v}/ c_{5} (\ell + 2)$
  hold.
On the other hand, if $b = 0$ then $\lamj{j}{q_{t}} \geq 1 = 2^{b}$
  for all $j \geq 1$, so $I (q_{t}, \lambda_{j} (q_{t})) \geq h_{v}/ c_{5} (\ell + 2)$
  trivially implies $j \geq 1$ and therefore $(C.3)$.

We will now finish the proof by contradiction: we show that if no
  such $t$ and $j$ exist, then the curve $I (q_{t_{h_{v}}}, \cdot)$
  would be too low.
If for all  $t \in ( t_{h_{v}-1}, t_{h_{v}}]$
  and all $j$ one of (C.1), (C.2) or 
  $I (q_{t}, \lambda_{j} (q_{t})) \geq h_{v}/ c_{5} (\ell + 2)$ fails,
  then Lemma~\ref{lem:C1help} tells us that for all $x$
\[
  I (q_{t_{h_{v}}}, x)
\leq
  \frac{h_{v}}{c_{5} (\ell + 2)} + \frac{3 x}{5 m}
+
  \sqrt{\widehat{x}} \left(1 - \frac{\phi^{2}}{2} \right)^{t_{1}}
\leq
  \frac{h_{v}}{c_{5} (\ell + 2)} + \frac{3 x}{5 m}
 +\frac{1}{c_{1} (\ell + 2)},
\]
by inequality~\eqref{eqn:tEnsure}.

In the case $x_{\ell} < 2 m / c_{6} (\ell +2)$,
  we obtain a contradiction by 
  plugging in $x = x_{h_{v}}$  to find
\[
  I (q_{t_{h_{v}}}, x_{h_{v}})
<
\frac{1}{\ell + 2}
\left(
 \frac{h_{v}}{c_{5}} + \frac{6}{5 c_{6}} + \frac{1}{c_{1}}
 \right)
\]
which by \eqref{eqn:c5leqc1c6} contradicts \eqref{eqn:lower2}.

In the case in that $x_{\ell} \geq  2 m / c_{6} (\ell +2)$,
  and so $h_{v} = \ell + 1$, we
  substitute $x = (2/3)2m$ to obtain
\[
  I (q_{t_{\ell +1}}, (2/3) 2m)
<
\frac{4}{5}
+
 \frac{1}{c_{5}}
 +
\frac{1}{c_{1} (\ell +2)},
\]
which by \eqref{eqn:manyc2} contradicts \eqref{eqn:lower1}.
\end{proof}

\subsection{Proof of Theorem~\ref{thm:Nibble}}

Fact (N.1) follows from conditions (C.1) and (C.2)
  in the algorithm.
Given a set $S$ satisfying $\vol{S} \leq (2/3) \vol{V}$,
  the lower bound on the volume of the set $S^{g}$
  is established in Lemma~\ref{lem:sizeSg}.
If $\Phi (S) \leq f_{1} (\phi )$ and $v \in S^{g}_{b}$,
  then Lemmas~\ref{lem:C123} and~\ref{lem:C4} show that
  the algorithm will output a non-empty set.
Finally, lemma~\ref{lem:N3} tells us that if $\Phi (S) \leq f_{1} (\phi)$,
  $v \in S^{g}$ and the algorithm outputs a non-empty set $C$,
  then it satisfies $\vol{C \intersect S} \geq 2^{b-1}$.

It remains to bound the running time of \texttt{Nibble}.
The algorithm will run for $t_{last}$
  iterations.
We will now show that with the correct implementation, each iteration
  takes time $O ((\log n)/ \epsilon)$.
Instead of performing a dense vector multiplication in step (3.a),
  the algorithm should keep track of the set of vertices $u$
  at which $r_{t} (u) > 0$.
Call this set $V_{t}$.
The set $V_{t}$ can be computed in time $O (\sizeof{V_{t}})$
  in step (3.b).
Given knowledge of $V_{t-1}$, the multiplication in step (3.a)
  can be performed in time proportional to 
\[
  \vol{V_{t-1}} =
\sum_{u \in V_{t-1}} d (u)
\leq
\sum_{u \in V_{t-1}} r_{t}(u) / \epsilon
\leq
1/ \epsilon .
\]
Finally, the computation in step (3.c) might require sorting the vectors
  in $V_{t}$ according to $r_{t}$, which could take time at most
  $O (\sizeof{V_{t}} \log n)$.
Thus, the run-time of \texttt{Nibble} is bounded by
\[
O \left(t_{last}
  \frac{\log n}{ \epsilon }
  \right)
=
O \left( t_{last}^{2}
  2^{b} \log^{2} m
  \right)
= O \left(
 \frac{2^{b} \log^{6} m}{\phi^{4}}
 \right).
\]

\section{Nearly Linear-Time Graph Partitioning}\label{sec:cut}

In this section, we apply \texttt{Nibble} to design a partitioning
  algorithm \texttt{Partition}.
This new algorithm runs in nearly linear-time.
It computes an approximate sparsest
  cut with approximately optimal balance.
In particular, we prove that there exists a constant $\alpha > 0$
  such that for any graph $G = (V,E)$ that has a cut $S$ of sparsity 
  $\alpha \cdot \theta^2/\log^3 n $
  and balance $b \leq 1/2$, with high probability, \texttt{Partition} finds a
  cut $D$ with $\Phi_{V} (D) \leq  \theta$  and
  $\balance{D} \geq b/2$.
Actually, \texttt{Partition} satisfies an even stronger
  guarantee: with high probability
  either the cut it outputs is well balanced,
\[
\frac{1}{4}\vol{V} \leq \vol{D}\leq
  \frac{3}{4}\vol{V},
\]
or touches most of the edges touching $S$,
\[
\vol{D \intersect S} \geq \frac{1}{2}\vol{S}.
\]
The expected running time of \texttt{Partition} is $O (m\log^{7} n/\phi^{4})$.
Thus, it can be used to quickly find crude cuts.

\texttt{Partition}
  calls \texttt{Nibble} via a routine
  called \texttt{Random Nibble} that calls \texttt{Nibble}
  with carefully chosen random parameters.
\texttt{Random Nibble}
  has a very small expected running time, and
  is expected to remove a similarly small fraction of
  any set with small conductance.

\subsection{Procedure \texttt{Random Nibble}}

\vskip 0.2in
\noindent
\fbox{
\begin{minipage}{6in}
\noindent $C = \mathtt{Random Nibble} (G, \phi )$
\begin{enumerate}
\item [(1)] Choose a vertex $v$ according to $\psi_{V}$.

\item [(2)] Choose a $b$ in $1, \ldots , \ceiling{\log m}$
  according to
\[
\prob{}{b = i}
 = 2^{-i} / (1-2^{-\ceiling{\log m}}).
\]

\item [(3)] $C = \mathtt{Nibble} (G, v, \phi , b)$.
\end{enumerate}
\end{minipage}
}
\vskip 0.2in

\begin{lemma}[\texttt{Random Nibble}]\label{lem:randNibble}
Let $m$ be the number of edges in $G$.
The expected running time of \texttt{Random Nibble} is
  $O \left(\log^{7} m / \phi^{4} \right) $.
If the set $C$ output by \texttt{Random Nibble} is non-empty, it satisfies
\begin{itemize}
\item [(R.1)] $\Phi_{V} (C) \leq \phi $, and
\item [(R.2)] $\vol{C} \leq (5/6)\vol{V}.$
\end{itemize}
Moreover,
  for every set $S$ satisfying
\[
\vol{S} \leq (2/3) \vol{V} \quad \mbox{and}\quad
 \Phi_{V} (S) \leq f_{1} (\phi ),
\]
\begin{itemize}
\item [(R.3)] $\expec{}{\vol{C \cap S}} \geq \vol{S} / 4 \vol{V}$.
\end{itemize}
\end{lemma}

\begin{proof}
The expected running time of \texttt{Random Nibble}
  may be upper bounded by
\[ O\left(
  \sum_{i=1}^{\ceiling{\log m}} \left(2^{-i}/ (1-2^{\ceiling{\log m}}) \right)
  \left(2^{i} \log^{6} (m) / \phi^{4} \right) \right)
=
O \left( \log^{7} (m) / \phi^{4} \right).
\]
Parts (R.1) and (R.2) follow directly from
  part (N.1) of Theorem~\ref{thm:Nibble}.
To prove part (R.3), define $\alpha_{b}$ by
\[
\alpha_{b}  = \frac{\vol{S^{g}_{b}}}{\vol{S^{g}}}.
\]

So, $\sum_{b}\alpha_{b} = 1$.
For each $i$, the chance that $v$ lands in
  $S^{g}_{i}$ is $\alpha_{i} \vol{S^{g}}/ \vol{V}$.
Moreover, the chance that $b = i$ is at least $2^{-i}$.
If $v$ lands in $S^{g}_{i}$, then by part (N.3) of
  Theorem~\ref{thm:Nibble}, $C$ satisfies
\[
 \vol{C \cap S} \geq  2^{i-1}.
\]
So,
\begin{align*}
  \expec{}{\vol{C \cap S}}
& \geq
  \sum_{i} 2^{-i} \alpha_{i}
  \left(\vol{S^{g}}/ \vol{V} \right)
  2^{i-1}\\
& =
  \sum_{i} (1/2) \alpha_{i}
  \left(\vol{S^{g}}/ \vol{V} \right)
  \\
& =
  \vol{S^{g}}/ 2 \vol{V}
\\
& \geq
  \vol{S}/ 4 \vol{V}.  \quad
  \text{(by part (N.2.a) of Theorem~\ref{thm:Nibble})}
\end{align*}
\end{proof}

\subsection{\texttt{Partition}}

We now define \texttt{Partition} and analyze its performance.
First, define
\begin{equation}\label{eqn:f2phi}
f_{2} (\theta)  \defeq f_{1} (\theta  / 7) /2,
\end{equation}
and note
\[
f_{2} (\theta) \geq \Omega \left(\frac{\theta ^{2}}{\log^{3} m} \right)
\]

\vskip 0.2in
\noindent
\fbox{
\begin{minipage}{6in}
\noindent $D =\mathtt{Partition}(G, \theta , p)$,
where $G$ is a graph, $\theta, p \in (0,1)$.
\begin{enumerate}
\item [(0)] Set $W_{0} = V$, $j = 0$ and $\phi = \theta /7$.
\item [(1)] While $j < 12 m \ceiling{\lg (1/p)}$ and
$\vol{W_{j}} \geq (3/4) \vol{V}$,
\begin{enumerate}
\item [(a)] Set $j = j + 1$.
\item [(b)] Set $D_{j} = \mathtt{Random Nibble} (G [W_{j-1}], \phi )$
\item [(c)] Set $W_{j} = W_{j-1} - D_{j}$.
\end{enumerate}
\item [(2)] Set $D = D_{1} \union \dotsb \union D_{j}$.
\end{enumerate}
\end{minipage}
}
\vskip 0.2in

\begin{theorem}[\texttt{Partition}]\label{thm:Partition}
On input a graph with $m$ edges,
  the expected running time of \mbox{\rm \texttt{Partition}} is 
  $O \left(m \lg (1/p) \log^{7} m / \theta^{4} \right) $.
Let $D$ be the output of $\mathtt{Partition} (G, \theta , p)$,
  where $G$ is a graph and $\theta, p \in (0,1)$.
Then
\begin{itemize}
\item [(P.1)] $\vol{D} \leq (7/8) \vol{V}$,
\item [(P.2)] If $D \not = \emptyset$ then $\conduc{V}{D} \leq  \theta $, and
\item [(P.3)]
If $S$ is any
  set satisfying
\begin{equation}\label{eqn:P3}
\vol{S} \leq \vol{V}/2 \quad \text{and} \quad  \conduc{V}{S} \leq f_{2} (\theta),
\end{equation}
  then with probability at least $1-p$,
  $\vol{D} \geq \vol{S}/2$.
\end{itemize}
In particular, with probability at least $1-p$ either
\begin{itemize}
\item [(P.3.a)] $\vol{D} \geq (1/4) \vol{V}$, or
\item [(P.3.b)] $\vol{S \intersect D} \geq \vol{S}/2$.
\end{itemize}
\end{theorem}

Property (P.3) is a little unusual and deserves some explanation.
It says that for every set $S$ of low conductance, with high probability
  either
  $D$ is a large fraction of $S$, or it is a large fraction of the
  entire graph.
While we would like to pick just one of these properties and guarantee that
  it holds with high probability, this would be unreasonable:  
  on one hand, there might be no big set $D$ of small conductance;
  and, on the other hand, even if $S$ is small the algorithm might
  cut out a large set $D$ that completely avoids $S$.

\begin{proof}[Proof of Theorem~\ref{thm:Partition}]
The bound on the expected running time of \texttt{Partition}
  is immediate from the bound on the running time of \texttt{RandomNibble}.

Let $j_{out}$ be the iteration at which \texttt{Partition} stops,
  so that $D = D_{1} \union \dotsb \union D_{j_{out}}$.
To prove (P.1), note that
  $\vol{W_{j_{out}-1}} \geq (3/4) \vol{V}$ and so
  $\vol{D_{1} \union \dotsb \union D_{j_{out}-1}} \leq (1/4) \vol{V}$.
By part (R.2) of Lemma~\ref{lem:randNibble},
  $\vol{D_{j_{out}}} \leq  (5/6) \vol{W_{j_{out}-1}}$.
So,
\[
\vol{D_{1} \union \dotsb \union D_{j_{out}}}
\leq 
\vol{V} - \vol{W_{j_{out}-1}} + \frac{5}{6} \vol{W_{j_{out}-1}} 
= 
\vol{V} - \frac{1}{6} \vol{W_{j_{out}-1}} 
\leq
 \frac{7}{8} \vol{V}.
\]

To establish (P.2),
  we first compute
\begin{align*}
\sizeof{E (D, V-D)}
& =
\sum_{i=1}^{j_{out}} \sizeof{E (D_{i}, V-D)}\\
& \leq
\sum_{i=1}^{j_{out}}
  \sizeof{E (D_{i}, W_{i-1} - D_{i})}\\
& \leq
\sum_{i=1}^{j_{out}}
  \phi \vol{D_{i}} \quad
\text{(by (R.1))}
\\
& = \phi  \vol{D}.
\end{align*}
So, if  $\vol{D} \leq \vol{V}/2$,
  then $ \conduc{V}{D} \leq \phi$.
On the other hand, we established above that
  $\vol{D} \leq (7/8) \vol{V}$, from which
  it follows that
\[
  \vol{V-D} \geq (1/8) \vol{V} \geq (8/7) (1/8) \vol{D} = (1/7) \vol{D}.
\]
So,
\[
  \conduc{V}{D}
=
 \frac{\sizeof{E (D, V-D)}}
      {\min \left(\vol{D}, \vol{V-D} \right)}
\leq
7  \frac{\sizeof{E (D, V-D)}}{\vol{D}}
\leq
7 \phi
=
\theta
\]

Let $j_{max} = 12 m \ceiling{\lg (1/p)}$.
To prove part (P.3), let $S$ satisfy \eqref{eqn:P3} and
  consider what happens if we ignore the second condition in the
  while loop and run
  \texttt{Partition} for all
  potential $j_{max}$ iterations, obtaining
  cuts $D_{1}, \dotsc , D_{12 m \ceiling{\lg (1/p)}}$.
Let 
\[
D^{\leq j} = \union_{i \leq j} D_{i}.
\]
We will prove that if neither 
\[
\vol{D^{\leq k}} \geq \frac{\vol{V}}{4}
\quad
\text{nor}
\quad
\vol{S \intersect D^{\leq k}} \geq \frac{\vol{S}}{2}
\]
hold at iteration $k$, then with probability at least $1/2$, one of these
  conditions will be satisfied by iteration $k + 12 m$.
Thus, after all $j_{max}$ iterations, one of conditions
  $(P.3.a)$ or $(P.3.b)$ will be satisfied with probability
  at least $1-p$.
If the algorithm runs for fewer iterations, then condition $(P.3.a)$
  is satisfied.

To simplify notation, let $C_{i} = D_{k+i}$ and $U_{i} = W_{k+i}$,
  for $0 \leq i \leq  12 m$.
Assume that
\[
\vol{U_{0}} \geq \frac{3}{4} \vol{V}
\quad \text{and} \quad 
\vol{S \intersect U_{0}} < \frac{1}{2} \vol{S}.
\]
For $1 \leq i \leq 12 m$, define the random variable
\[
  X_{i} = \frac{\vol{C_{i} \intersect S}}{\vol{U_{0} \intersect S}}.
\]
As each set $C_{i}$ is a subset of $U_{0}$, and the $C_{i}$ are mutually
  disjoint, we will always have
\[
  \sum_{i = 1}^{12m} X_{i} \leq 1.
\]
Define $\beta$ to satisfy
\[
  (1-\beta) \vol{S \intersect U_{0}} = \frac{1}{2} \vol{S},
\]
and note that this ensures $0 < \beta \leq 1/2$.
Moreover, if $\sum X_{i} \geq \beta$, then 
  $\vol{S \intersect D^{\leq k+12m}} \geq \vol{S}/2$
  will hold.

Let $E_{j}$  be the event
\[
\vol{U_{j}} < \frac{3}{4} \vol{V}.
\]
We need to show that, with probability at least $1/2$, either an
  event $E_{j}$ holds, or $\sum X_{i} \geq \beta$.
To this end, we now show that if neither $E_{j}$
  nor $\sum_{i \leq j} X_{i} \geq \beta$ holds,
  then $\expec{}{X_{j+1}} \geq 1/8 m$.
If $\sum_{i \leq j} X_{i} < \beta$, then
\[
\vol{S \intersect U_{j}}
= 
\vol{S \intersect U_{0}}
-
\sum_{i \leq j} \vol{S \intersect C_{i}}
=
\vol{S \intersect U_{0}}
\left(1 - \sum_{i \leq j} X_{i} \right)
>
\vol{S \intersect U_{0}} (1 - \beta )
= 
\frac{1}{2} \vol{S}.
\]
If $E_{j}$ does not hold, then
\[
  \vol{U_{j} - S \intersect U_{j}}
= 
  \vol{U_{j}} - \vol{S \intersect U_{j}}
\geq 
\frac{3}{4} \vol{V} - \vol{S}
\geq 
\frac{1}{4} \vol{V}
\geq 
\frac{1}{2} \vol{S}.
\]
So,
\[
\conducin{U_{j}}{S \intersect U_{j}}
= 
\frac{
  \sizeof{E (S \intersect U_{j}, U_{j} - S \intersect U_{j})}
}{
 \min \left(\vol{S \intersect U_{j}}, \vol{U_{j} - S \intersect U_{j}} \right)
}
\leq 
\frac{
  \sizeof{E (S , V - S)}
}{
  (1/2) \vol{S}
}
\leq 
2 \conduc{G}{S}
\leq 
2 f_{2} (\theta)
=
f_{1} (\phi).
\]
We also have $\vol{S \intersect U_{j}} \leq \vol{S} \leq (2/3) \vol{U_{j}}$,
 so the conditions of 
 part $(R.3)$ of Lemma~\ref{lem:randNibble} are satisfied and
\[
  \expec{}{X_{j+1}} \geq 1 / 4 \vol{U_{j}} \geq 1 / 8 m.
\]

Now, set 
\[
  Y_{j} = 
\begin{cases}
1 / 8m & \text{if $\sum_{i < j} X_{i} \geq \beta$, or if $E_{j-1}$}          
\\
X_{j} & \text{otherwise}.
\end{cases}
\]
So, for all $j$ we have $\expec{}{Y_{j}} \geq 1 / 8m$, 
  and so $\expec{}{\sum_{j \leq 12 m} Y_{j}} \geq 3/2$.
On the other hand,
  
\[
\sum_{i \leq 12 m} Y_{j} \leq \sum X_{i} + \frac{12 m}{8 m} \leq 5/2.
\]
So, with probability at least $1/2$, 
\[
  \sum_{j \leq 12 m} Y_{j} \geq 1/2 \geq \beta .
\]
This implies that with probability at least $1/2$ either
  $\sum_{i} X_{i} \geq \beta$ or some event $E_{j}$ holds,
  which is what we needed to show.
\end{proof}

\bibliographystyle{alpha}
\bibliography{precon}

\end{document}